\documentclass[acmsmall]{acmart}
\usepackage{standalone}
\usepackage{placeins}

\newcommand\shorttitle{\relax}
\usepackage{import}
\usepackage[T1]{fontenc}
\usepackage[utf8]{inputenc}

\usepackage{amsmath}
\usepackage{amssymb}
\usepackage{amsthm}

\usepackage{xcolor}
\usepackage{hyperref}
\hypersetup{
  hidelinks,
  pdfcreator={LaTeX via pandoc}}
\usepackage{listings}
\lstdefinelanguage{scalaish}{
  basicstyle=\small\ttfamily,
  keywords={val, if, then, in, handle, return, def, match, case, new, type, effect, pretype, class, extends, infix, else},
  keywordstyle=\bfseries,
  sensitive=false,
  comment=[l]{//},
  commentstyle=\color{gray},
  stringstyle=\color{gray}, 
  morestring=[b]',
  morestring=[b]"
}

\usepackage{bcprules}
\usepackage{mathpartir}
\usepackage{multicol}
\usepackage{placeins}

\renewcommand{\comment}[1]{}
\newcommand\ie{\emph{i.e.,}~}

\newcommand{\fsub}{$\textsf{System F}_{<:}~$}
\newcommand\CC{$\textsf{CF}_{<:}~$}

\newcounter{RuleRef}
\newcommand{\ruledef}[1]{%
  \refstepcounter{RuleRef}\textsc{#1}\label{rule:#1}}
\newcommand{\ruledefN}[2]{%
  \refstepcounter{RuleRef}\textsc{#2}\label{rule:#1}}

\newcommand{\ruleref}[1]{\textsc{(\hyperref[rule:#1]{#1})}}
\newcommand{\rulerefN}[2]{\textsc{(\hyperref[rule:#1]{#2})}}

\newcommand{\figurebox}[1]
        {\small\fbox{\begin{minipage}{\textwidth} #1 \medskip\end{minipage}}}

\newcommand{\boxfig}[3]
        {\begin{figure*}\figurebox{#3\caption{\label{#1}#2}}\end{figure*}}
\newcommand{\figref}[1]
        {Figure~\ref{#1}}

\newcommand{\judgement}[2]{{\bf\textsf{#1}} \hfill #2}

\newcommand\kw[1]{\operatorname{\textbf{\textsf{#1}}}}

\definecolor{light-gray}{gray}{0.92}
\newcommand{\highlight}[1]{\colorbox{light-gray}{$\displaystyle #1$}}

\renewcommand\Top{\top}
\newcommand{\Double}{\id{Double}}
\newcommand\Unit{\id{Unit}}
\newcommand{\List}[1]{\id{List}[#1]}

\newcommand\lam[3]{\lambda\left(#1 : #2\right)\!.\; #3}
\newcommand\lambr[3]{\lambda\left(#1 : #2\right)\!.\;\\&#3}
\newcommand\tlam[3]{\Lambda\left[#1 <: #2\right]\!.\; #3}
\newcommand\tlambr[3]{\Lambda\left[#1 <: #2\right]\!.\;\\&#3}

\newcommand\LAM[1]{\lambda\left(#1\right)\!.\;}
\newcommand\TLAM[1]{\forall(#1) \rightarrow}
\newcommand\Tlam[3]{\TLAM{#1 : #2} #3}

\newcommand\Ttlam[3]{\forall\left[#1 <: #2\right] \rightarrow #3}
\newcommand\dom[1]{\operatorname{dom} #1}

\newcommand\cset[1]{\{#1\}}
\newcommand\UC{\cset{*}}
\newcommand\emptycs{\cset{}}
\newcommand\capt[1]{\cset{#1} \;}
\newcommand\CAPT[2]{\cset{#1} \; #2}
\newcommand\Capt[2]{#1 \; #2}
\newcommand\app[2]{#1 \; #2}
\newcommand\tapp[2]{#1\left[#2\right]}

\newcommand\subs[3]{\left[#1 \mapsto #2\right]#3}

\newcommand{\ts}{\,\vdash\,}
\newcommand\sub{<:}
\newcommand\nts{\,\nvdash\,}
\newcommand\typ{:}
\newcommand{\ccomma}{,\;}
\newcommand{\csemi}{\;;\;}

\newcommand\hole{\square}

\newcommand\fv{\operatorname{fv}}
\newcommand\cv{\operatorname{cv}}

\newcommand\id[1]{\operatorname{\mathsf{#1}}}
\newcommand\local[1]{\operatorname{\mathit{#1}}}

\newcommand{\IF}{\kw{if}}
\newcommand{\THEN}{\kw{then}}
\newcommand{\ELSE}{\kw{else}}
\newcommand{\LET}{\kw{let}}
\newcommand{\BOX}{\kw{box}}
\newcommand\IN{\;\mathop{\kw{in}}\;}
\newcommand\HANDLE{\kw{handle}}

\newcommand\wf{\;\mbox{$\kw{wf}$}}

\newcommand{\reduces}{\;\longrightarrow\;}

\newcommand\Return{\id{Return}}
\newcommand\RETURN[1]{#1.\!\kw{return}\;}

\newcommand\Eff{\id{Eff}}
\newcommand\E{\textsf{E}}

\newcommand\eff[1]{#1.\kw{do}}

\newcommand\region[2]{\kw{region} #1 \kw{in} #2}
\newcommand\new[2]{#1.\!\kw{new}[#2]}
\newcommand\Region{\id{Region}}
\newcommand\Ptr{\id{Ptr}}

\newcommand{\gap}{\quad\quad}

\newcommand{\lindent}{\hspace{-4mm}}

\newcommand{\bbox}[1]{\mbox{\textbf{\textsf{#1}}}}

\acmConference{Technical Report}{2021}{arXiv}
\acmYear{2021}
\acmISBN{} 
\acmDOI{} 
\startPage{1}

\setcopyright{none}

\makeatletter
\let\@authorsaddresses\@empty
\makeatother

\begin{document}

\author{Aleksander Boruch-Gruszecki}
\affiliation{
  \department{LAMP – Programming Methods Laboratory}              
  \institution{EPFL}            
  \streetaddress{Station 14}
  \city{Lausanne}
  \postcode{1015}
  \country{Switzerland}                    
}
\email{aleksander.boruch-gruszecki@epfl.ch}          

\author{Jonathan Immanuel Brachth\"auser}
\affiliation{
  \department{LAMP – Programming Methods Laboratory}              
  \institution{EPFL}            
  \streetaddress{Station 14}
  \city{Lausanne}
  \postcode{1015}
  \country{Switzerland}                    
}
\email{Jonathan.Brachthauser@epfl.ch}          

\author{Edward Lee}
\affiliation{
  \department{School of Computer Science}              
  \institution{University of Waterloo}            
  \streetaddress{200 University Ave W.}
  \city{Waterloo}
  \state{ON}
  \postcode{N2L 3G1}
  \country{Canada}                    
}
\email{e45lee@uwaterloo.ca}          

\author{Ond\u{r}ej Lhot\'ak}
\affiliation{
  \department{School of Computer Science}              
  \institution{University of Waterloo}            
  \streetaddress{200 University Ave W.}
  \city{Waterloo}
  \state{ON}
  \postcode{N2L 3G1}
  \country{Canada}                    
}
\email{olhotak@uwaterloo.ca}          

\author{Martin Odersky}
\affiliation{
  \department{LAMP – Programming Methods Laboratory}              
  \institution{EPFL}            
  \streetaddress{Station 14}
  \city{Lausanne}
  \postcode{1015}
  \country{Switzerland}                    
}
\email{martin.odersky @epfl.ch}          


\title{Tracking Captured Variables in Types}

\begin{abstract}
  Type systems usually characterize the shape of values but not their free variables.
  However, there are many desirable safety properties one
  could guarantee if one could track how references can escape.
  For example, one may implement {\it algebraic effect handlers} using
  {\it capabilities} -- a value which {\it permits} one to perform
  the effect -- safely if one can guarantee that the
  capability itself does not escape the scope bound by the effect handler.
  To this end, we study the \CC calculus, a conservative and
  lightweight extension of $\textsf{System F}_{<:}~$, to track how values and
  their references can be captured and escape. We show that existing terms in
  \fsub embed naturally in our calculus, and that many natural problems can be
  expressed in a system that tracks variable references like we do in
  \CC. We also give mechanized proofs of the soundness
  properties of \CC in Coq. The type system presented in
  \CC is powerful enough to reason about safety in the context of
  many natural extensions of \CC such as {\it region-based
    memory-management}, {\it non-local returns}, and {\it effect handlers}.
\end{abstract}

\maketitle

\renewcommand{\shortauthors}{Boruch-Gruszecki, Brachthäuser, Lee, Lhot\'ak, and Odersky}

\hypertarget{the-capture-calculus}{%
\section{Introduction}\label{the-capture-calculus}}

Computing the \emph{free variables} of a term is one of the most basic operations that
students of programming language theory are exposed with. Yet, it has significant
relevance, not only in meta-theory -- but as we will study in this paper -- also as
a programming device. In particular, combined with an object-capability discipline \citep{miller2006robust}
the free variables of a term inform us about the \emph{authority} of this term.
In general, free variables can be used to express \emph{global} capabilities, restricting access
to privileged operations (like using FFI, accessing the network, reading, writing to files, etc.)
to the holders of the corresponding capabilities.
They also can be used to phrase \emph{effect safety} in terms of capability safety:
to establish effect safety, it is important to guarantee that \emph{local} capabilities,
introduced by exception (or effect) handlers, do not leave the corresponding handler.
One particular problem related to analyzing whether a capability escapes is \emph{capture},
that is, function values closing over capabilities. By means of capture,
a capability can indirectly (and potentially unnoticed) flow to some other component, transferring the privileges.

Motivated by the above mentioned use cases, in this paper we internalize the concept
of free variables and introduce \CC, a calculus equipped with a type system
based on the idea to track the \emph{free variables of a value} in its type,
thereby making capture visible.
\CC builds on \fsub and enriches its types to allow tracking
captured variables.

\paragraph{Tracking variables in capture sets}
Specifically, we make two significant additions.
First, we introduce a notion of \emph{tracked} variables to represent resources,
capabilities, and other information that should be tracked by the type system.
Second, we augment \fsub~ types with \emph{capture sets} $\{x_1, \hdots x_n\}$.
Terms of the type $\Capt{\{x_1, \hdots, x_n\}}{U}$ represent expressions of type $U$
whose reduced values may only refer to (\ie capture) tracked variables
in the set $\{x_1, \hdots, x_n\}$.
These concepts are illustrated in the following example.
\definecolor{capabilitycolor}{RGB}{39, 121, 156}
\definecolor{capabilitybg}{RGB}{225, 239, 245}
\definecolor{hlbg}{RGB}{255, 179, 102}
\definecolor{hlcolor}{RGB}{179, 59, 0}
\definecolor{stringcolor}{RGB}{0, 130, 91}
\newcommand{\capability}[1]{\colorbox{capabilitybg}{{\color{capabilitycolor}{$\sf{#1}$}}}}
\newcommand{\Logger}[0]{{\sf{Logger}}}
\newcommand{\hl}[1]{\colorbox{hlbg}{{\color{hlcolor}{$\sf{#1}$}}}}
\newcommand{\filecap}[0]{\capability{File}}
\newcommand{\consolecap}[0]{\capability{Console}}
\newcommand{\str}[1]{{\color{stringcolor}{\texttt{"#1"}}}}
\newcommand\call[2]{#1.{\textbf{#2}}\;}
\newcommand{\lineVar}{\textit{line}}
\newcommand{\logVar}{\textit{log}}
{\small
\begin{align*}
&\id{fileLogger} : \CAPT{\filecap} \Logger\\
&\id{fileLogger} = \LAM{\lineVar} \call{\filecap}{append} \str{log.txt}\; \lineVar\
\end{align*}
}

\noindent
Here the function type $\Logger = {\sf{String}} \rightarrow {\sf{Unit}}$ is annotated with a capture set $\cset{\filecap}$ making
visible in the type that the body of function $\id{fileLogger}$ closes over the global capability $\filecap$, which is a tracked variable.
In the same way, we can define alternative logger implementations that close over different capabilities:

{%
\begin{minipage}[t]{0.5\textwidth}
{\small%
\begin{align*}
&\id{printLogger} : \CAPT{\consolecap} \Logger\\
&\id{printLogger} = \LAM{\lineVar} \call{\consolecap}{println} \lineVar
\end{align*}}
\end{minipage}
\begin{minipage}[t]{0.5\textwidth}
{\small%
\begin{align*}
&\id{pureLogger} : \CAPT{} \Logger\\
&\id{pureLogger} = \LAM{\lineVar} ()
\end{align*}}
\end{minipage}
}

\paragraph{Capture polymorphism}
For additional expressivity, our calculus also supports some form of \emph{capture polymorphism}. That is,
variables bound by lambda abstractions can be used in types to refer to the free variables (the capture set)
of the evaluated argument.

{%
\begin{minipage}[t]{0.5\textwidth}
{\small
\begin{align*}
  &\id{warn} : \CAPT{} \TLAM{\logVar : \UC \; \Logger} \CAPT{\logVar} \Logger\\
  &\id{warn} = \LAM{\logVar} \LAM{\lineVar} \logVar \;(\str{[WARN]} + \lineVar)
\end{align*}}
\end{minipage}
\begin{minipage}[t]{0.5\textwidth}
{\small
\begin{align*}
  &\id{myLogger} : \CAPT{\consolecap} \Logger\\
  &\id{myLogger} = \id{warn}\; \id{printLogger}
\end{align*}}
\end{minipage}
}

The type of $\mathsf{warn}$ reads as ``given an argument logger $\logVar$, with an unknown capture set $\UC$ the returned value of type $\Logger$ may close over $\logVar$''. The type of the function shows that we introduce a simple form of term dependency.
For the reader's convenience, we visually distinguish capabilities (like
$\filecap$) from variables (like $\logVar$). The former will remain free
under reduction while the latter will eventually be
substituted away in capture sets, as can be seen in the type of
$\mathsf{myLogger}$. There, passing $\mathsf{printLogger}$ to the
capture polymorphic function $\mathsf{warn}$ substitutes $\logVar$ in the result
of $\mathsf{warn}$ with $\cset{\consolecap}$, resulting in
$\CAPT{\consolecap} \Logger$.

\paragraph{Subcapturing}
Building on \fsub, our calculus enables subtyping on capture sets, which we refer to as \emph{subcapturing}.
In our example, we have that $\CAPT{} \Logger \sub \CAPT{\filecap} \Logger$ since $\emptycs$ is a subset of $ \cset{\filecap}$.

\paragraph{Capture prediction}
From a programmers perspective, the capture set $C$ on a function type like $C \; \Logger$
provides us with an \emph{upper bound} on the free variables of values of this type (Corollary \ref{lemma:predict-term}).
That is, the function body can only use those capabilities explicitly passed to the function and those mentioned in $C$.
For example, the type $\CAPT{\consolecap} \Logger$ informs us that $\mathsf{printLogger}$ might at most use $\consolecap$,
but not (for example) access files by means of the $\filecap$ capability.
Capture prediction equips us with knowledge about the capture of \emph{values}, not that of arbitrary terms.
The difference is illustrated in the following example term.
{\small
  \begin{align*}
    &\id{someLogger} : \CAPT{\highlight{?}} \Logger\\
    &\id{someLogger} = \kw{if}\; (\LAM{\logVar} \kw{true})\;\id{consoleLogger} \;\kw{then}\; \id{printLogger} \;\kw{else}\; \id{pureLogger}
  \end{align*}}

\noindent
The term $\id{someLogger}$ will either reduce to $\id{printLogger}$ or $\id{pureLogger}$, but
its definition mentions $\id{consoleLogger}$ as well. What should the highlighted capture
set of $\id{someLogger}$ be? In \CC the capture set on a type predicts the free variables of
the value that term reduces to. In our example, we can type $\id{someLogger}$ with the following type
{\small
  \begin{align*}
    \id{someLogger} : \CAPT{\filecap} \Logger
  \end{align*}}
\noindent
since both branches can be typed against $\CAPT{\filecap} \Logger$. The fact that
the condition also refers to $\id{consoleLogger}$ is irrelevant for the typing of the returned value.
Importantly, this correctly allows us to predict that $\id{warn}\;\id{someLogger}$ cannot possibly reference
$\consolecap$.

\paragraph{Applications}
While the above examples can provide a good first intuition, it is important
to note that simply harnessing the power of free variables, our calculus is
completely parametric in the semantics of global capabilities.
In general, being able to predict the free variables
of the value that a given term reduces to, we are able to develop soundness arguments for
the following applications:
\begin{enumerate}
  \item {\it Safe algebraic effects:} One can add safe algebraic effects to a purely functional core language
    by modelling them as {\it capabilities}. {\it Capabilities} are regular values that are introduced by
    special program constructs. For a concrete example, consider the algebraic effect of {\it throwing an exception}.
    In the case of exceptions, the capability to raise an exception could be
    introduced by a \emph{try} handler. We would like to
    ensure that exceptions can be raised only when they are handled by
    an enclosing \emph{try}. This means we need to make sure that the \emph{"can-raise-exception"} capability (which is
    a regular value) cannot escape the scope of the \emph{try} as a free variable in its result value.
    The type system presented in this paper can be used to enforce such a constraint, as we show
  in Sections \ref{non-local-returns} and \ref{general-algebraic-effects}.
  \item {\it Regions:} A {\it region} is a lexically delimited scope in which values
    can be allocated.  One concrete example would be a local variable to a function with
    stack-allocated local values.  After the region is exited,
    in order to be sound, one needs to ensure that there are no {\it dangling references} to values that were allocated within
    the region.  \CC can be used to enforce this restriction in order to ensure soundness,
    as we show in Section \ref{region-extension}.
\end{enumerate}
Moreover, there are many other applications which can be shown sound using similar arguments
to the ones we have presented in this work.  For example, ensuring that a {\it handle} to a resource
does not leak after it has been {\it closed} is a very similar problem to ensuring that {\it references}
to a stack allocated value in a {\it region} do not leak after the region has been {\it deallocated}.
In the remainder of the paper, we introduce the calculus \CC and illustrate its use. In particular, as we will see,
while the idea of tracking free variables in the type appears very intuitive, the interaction with subtyping proved to be
challenging and required several iterations of careful tradeoffs.

\subsection{Contributions}
Concretely, this paper makes the following contributions.

\begin{itemize}
  \item[--] We develop type-theoretic foundations of tracking free variables,
            resulting in a new calculus \CC (Section \ref{sec:cc-formally}).
            The calculus enhances types with additional information about variable
            capture, recorded in \emph{capture sets}. Subset inclusion of capture sets
            immediately motivates the need for subtyping. In consequence, we formalize \CC as a modest
            extension to \fsub.

  \item[--] We prove the standard soundness theorems (Section \ref{metatheory}).
            We capture the essence of \CC in Corollary \ref{lemma:predict-term}, which
            shows that capture sets are meaningful and provide a conservative approximation
            of the free variables of a value. The paper is accompanied by a fully mechanized
            soundness proof using the Coq theorem prover.

  \item[--] We show the applicability of \CC to a wide range of interesting applications,
            including systems describing regions, effects, or capabilities (Section \ref{sec:extensions}).
            In various extensions to the calculus, we make use of the fact that capture sets are meaningful,
            which implies that they can be used as a sound mechanism to prevent variables from escaping.
\end{itemize}

\section{The \CC Calculus}
\label{sec:cc-formally}

In this section, we formally introduce \CC, which allows us to discuss important meta theoretic aspects,
such as soundness (Theorems \ref{th:progress} and \ref{th:preservation}) and capture prediction (Corollary \ref{lemma:predict-term}). The core calculus presented in this language merely provides all necessary means to track free variables in types. In Section \ref{sec:extensions}, we extend \CC with additional features that put the tracking into use.

\hypertarget{abstract-syntax}{%
\subsection{Syntax of Terms and Types}\label{abstract-syntax}}

Figure \ref{fig:syntax} defines the syntax of \CC. Our language builds on \fsub with the following changes:

\subsubsection*{Types and Pretypes} We make a distinction between {\it pretypes} $U$
  and {\it types} $T$. Each type has a single capture set associated with it. In
  contrast, {\it pretypes} are ``incomplete'' types not yet associated with a capture set.
  All \fsub types save for type variables $X$ are pretypes in our
  calculus. A type variable stands for a complete type and, accordingly, is not
  a pretype. As usual, typing contexts $\Gamma$ can contain
  both term bindings $x : T$ and type bindings $X <: T$.
\subsubsection*{Capture Sets} Values of type $C \; U$ can be
  viewed as values of type $U$ that might contain occurrences of variables in
  $C$. Capture sets are $C$ are either a finite set of variables or the special
  set $\UC$ that conceptually represents a set containing every variable. Values
  of the type $\CAPT{} U$ are {\it pure} as they cannot capture tracked
  variables.

\subsubsection*{Function types in \CC are dependent} Function types in our calculus have a
fundamental difference compared to their kin in \fsub; instead of $S \rightarrow T$, we write
$\TLAM{x : S} T$, where $x$ names the bound parameter. This binding is needed
since $x$ may be used as a variable in the capture sets embedded in $T$.
Note that capture tracking is the only form of term-dependency in our calculus.

\noindent Observe that our core calculus does not have any base capabilities and
does not even distinguish syntactically between variables and capabilities. In
\ref{sec:extensions}, we will demonstrate that it is possible to extend the core
we present here with capabilities by treating them as variables.

\boxfig{fig:syntax}{The syntax of the \CC calculus.}{
  \judgement{\textsf{Syntax}}{}

  \begin{center}
  \[
  \begin{array}[t]{l@{\hspace{10mm}}l}

  x, y & \lindent{\bbox{Term Variable}} \\

  C ::= & \lindent{\bbox{Capture set}} \\
  \gap \{x_1, x_2, \ldots, x_n\} & \mbox{concrete capture set} \\
  \gap \UC & \mbox{universal capture set} \\

  \Gamma ::= & \lindent{\bbox{Environment}} \\
  \gap \cdot & \\
  \gap \Gamma \ccomma x : T & \\
  \gap \Gamma \ccomma X <: T & \\

  v, w ::= & \lindent{\bbox{Value}} \\
  \gap \LAM{x : T} t & \mbox{term abstraction} \\
  \gap \tlam{X}{S}{t} & \mbox{type abstraction} \\

  s, t ::= & \lindent{\bbox{Term}} \\
  \gap v & \mbox{value} \\
  \gap x & \mbox{variable} \\
  \gap \app{s}{t} & \mbox{application} \\
  \gap \tapp{t}{T} & \mbox{type application} \\

  R, S, T ::= & \lindent{\bbox{Type}} \\
  \gap X, Y & \mbox{type variable} \\
  \gap \Capt{C}{U} & \mbox{U capturing C} \\

  U, V, W ::= & \lindent{\bbox{Pretype}} \\
  \gap \Top & \mbox{top type} \\
  \gap \TLAM{x : S} T & \mbox{term function} \\
  \gap \Ttlam{X}{S}{T} & \mbox{type function} \\

  \end{array}
  \]
  \end{center}

}

\hypertarget{preliminaries}{%
\subsection{Preliminaries}\label{preliminaries}}

\begin{enumerate}
\item The universal capture set $\UC$ conceptually represents a set of all tracked variables.
  Set union and set difference are extended to the universal capture set $\UC$ as follows:
\[
\begin{array}[t]{l}
\UC \cup C = C \cup \UC = \UC \\
\UC \setminus C = \UC \\
\end{array}
\]
\item Substitution of capture sets $\subs{x}{C_1}C_2$ is defined as
follows:

\[
\begin{array}[t]{l@{\hspace{2mm}}l@{\hspace{4mm}}l}
  \subs{x}{C_1}{C_2} &= (C_2 \setminus \{x\}) \cup C_1 & \operatorname{if} x \in C_2 \\
                     &= C_2 &\operatorname{otherwise}
\end{array}
\]

Substitution is lifted as a homomorphism to types and pretypes, with

\[
 \subs{x}{C_1}{(C_2 \; U)} = \subs{x}{C_1}{C_2} \; \subs{x}{C_1}{U}
\]

and to terms $\subs{x}{C}{t}$, substituting capture sets in type positions.

\item The free variables $\fv(t)$ of a term $t$ only consider variables in term position;
  they do not include variables that are free but only occur in a capture set in a type
  which occurs in $t$.

\item The capture set $\cv(T, \Gamma)$ of a type
$T$ in a context $\Gamma$ is defined as follows:

\[
\begin{array}[t]{l@{\hspace{2mm}}l@{\hspace{2mm}}l}
  \cv(\Capt{C}{U}, \Gamma) &= C \\
  \cv(X, \Gamma) &= \cv(T, \Gamma) & \mbox{$\operatorname{if} X <: T \in \Gamma$.}
\end{array}
\]

\end{enumerate}

\hypertarget{evaluation}{%
\subsection{Evaluation}\label{evaluation}}

Evaluation in \CC{} is almost exactly the same as in
call-by-value \fsub{}. \figref{fig:reduction} defines the operational semantics
with a
single congruence rule that takes an evaluation context $E$.
The only major change to the reduction semantics in \CC{} compared to \fsub{} is that
reducing a term application $\app{\left(\LAM{x : T} t\right)}{v}$ with
\ruleref{beta-v}, we also need to substitute the occurrences of the lambda parameter
$x$ in capture set positions inside $t$. A value captures exactly the free
variables it references, so we substitute with $\fv(v)$.
The calculus we present is specialised for call-by-value semantics, as we can
see in Lemma \ref{lem:subst_ee'} -- term substitution preserves typing only if we
substitute with values. To see why, recall the $\id{warn}\;\id{someLogger}$
example from the introduction -- the type we assigned to this term took into
account that $\id{someLogger}$ will be reduced before the substitution. If
desired, the typing rules of \CC could be adjusted to account for the fact that
in call-by-name semantics, function application can capture more than in
call-by-value.

\boxfig{fig:reduction}{Small step operational semantics of the \CC calculus.}{
  \judgement{Evaluation}{\fbox{$t \reduces t$}}

  \infax[\ruledefN{beta-v}{beta-v}]{%
    \app{\left(\lam{x}{T}{t}\right)}{v} \reduces \subs{x}{v}{\highlight{\subs{x}{\fv(v)}{t}}}
  }

  \infax[\ruledefN{beta-T}{beta-T}]{%
    \tapp{\left(\tlam{x}{S}{t}\right)}{T} \reduces \subs{x}{T}{t}
  }

  \infrule[\ruledef{context}]{%
    t_1 \reduces t_2
  }{%
    \E[t_1] \reduces \E[t_2]
  }


  \begin{center}
    $
    \begin{array}[t]{l@{\hspace{10mm}}l}
      \E ::=\; \hole \;\mid\; \app{\E}{t} \;\mid\; \tapp{\E}{T} \;\mid\; \app{v}{\E} &  \lindent{\bbox{Evaluation context}} \\
    \end{array}
    $
  \end{center}
}

\hypertarget{subcapturing-rules}{%
\subsection{Subcapturing Rules}\label{subcapturing-rules}
} Subcapturing $\Gamma \ts C_1 \sub C_2$ is defined on capture sets as shown in
\figref{fig:subcapturing} -- both transitivity and reflexivity are
admissible. If one set subsets another, it also subcaptures it, but the opposite
is not necessarily true. The reason for that is that our capture sets are {\it
  indirect}.
For instance, under a variable binding $x: \CAPT{y} \Logger$, if a term
captures $x$ then intuitively it also indirectly captures $y$. Such a term will
have the capture set $\{x\}$, not $\{x,y\}$. Under such a binding, we would be
able to deduce that $\{x\} \sub \{y\}$, using rule \ruleref{sc-var}.
However, note that the converse is not true: we do not have $\{y\} \sub \{x\}$
as $\{y\}$ is not as precise; $x$ may be instantiated with a pure value which
can only capture pure values.
In general, a term with a type of the form $\CAPT{x} U$ can capture {\it no
  more} than $x$ -- however, it can potentially capture less.
In other words, capture set ascriptions on lambda parameters are upper
bounds on what the actual argument may capture.
These two notions -- indirect capture sets and capture sets being
only upper bounds -- are what enables our approach to capture
polymorphism. Recall the $\mathsf{warn}$ function:
{\small
  \begin{align*}
    &\id{warn} : \CAPT{} \TLAM{\logVar : \UC \; \Logger} \CAPT{\logVar} \Logger\\
    &\id{warn} = \LAM{\logVar} \LAM{\lineVar} \logVar \;(\str{[WARN]} + \lineVar)
  \end{align*}}
\noindent It is also possible to type $\id{warn}$ as
${\CAPT{} \TLAM{\logVar : \UC\;\Logger} \UC\;\Logger}$ -- indeed, if our capture
sets directly contained all their transitive members, this would be the only
logical choice. However, by doing that we would lose the type-level knowledge that
the result of $\id{warn}$ captures no more than its argument. And if we are to
exploit this knowledge, we obviously must also allow arguments to be typechecked
with capture sets smaller than $\UC$.

Finally, rule \ruleref{sc-var} also allows pure variables to be dropped from
capture sets; if $x$ of type $T$ is pure, this means
$\cv(T, \Gamma) = \emptycs$, hence $\{x\}$ is in a subcapturing relation with
any capture set, including the empty set. For a concrete example, we can derive
$ y : \{\} \Top, x : \UC \Top \ts \{x,y\} \sub \{x\}$.

\comment{
What is a practical consequence of our capture sets being indirect?

Assume $\Gamma \triangleq \id{printLogger} : \CAPT{\consolecap} \Logger$.

Then $\Gamma \nvdash  \{\consolecap\} <: \{\id{printLogger}\}$.

Take $t \triangleq \LAM{x} \id{printLogger} \; \str{msg}$.

Then, $\Gamma \ts t : \CAPT{\id{printLogger}} \Unit$.
}

\boxfig{fig:subcapturing}{Subcapturing of capturesets in the \CC calculus.}{
  \judgement{Subcapturing}{\fbox{$\Gamma \ts C \sub C$}}

  \infax[\ruledef{sc-*}]{%
    \Gamma \ts C \sub \UC}

  \infrule[\ruledef{sc-distl}]{%
    \Gamma \ts \cset{x_1} \sub C, \ldots, \Gamma \ts \cset{x_n} \sub C
  }{%
    \Gamma \ts \cset{x_1, \ldots, x_n} \sub C}

  \infax[\ruledef{sc-distr}]{%
    \Gamma \ts \cset{x_i} \sub \cset{x_1, x_2, \ldots, x_n}
  }

  \infrule[\ruledef{sc-var}]{%
    \Gamma \ts x : T \gap \Gamma \ts \cv(T, E) \sub C
  }{%
    \Gamma \ts \cset{x} \sub C}

%
}

\hypertarget{subtyping-rules}{%
\subsection{Subtyping Rules}\label{subtyping-rules}}

Due to the type/pretype split, there are technically two subtyping judgements,
as shown in Figure \ref{fig:subtyping}; one for types with rules \ruleref{capt}
and \ruleref{tvar} and one for pretypes with rules \ruleref{fun},
\ruleref{tfun}, and \ruleref{top}. Reflexivity and transitivity apply to each
kind of judgement; they are the only duplicated rules.
Note that the subtyping rules are a straightforward extension of the subtyping rules
for \fsub; the only significant departure is the addition of \ruleref{capt}
for reasoning with capture sets in types.

\boxfig{fig:subtyping}{Subtyping of types (and pretypes, correspondingly) in the \CC calculus.}{
  \judgement{Subtyping}{\fbox{$\Gamma \ts T \sub T$}}

  \infax[\ruledef{refl-type}]{%
    \Gamma \ts T \sub T}

  \infrule[\ruledef{trans-type}]{%
    \Gamma \ts R \sub S \gap \Gamma \ts S \sub T
  }{%
    \Gamma \ts R \sub T}

  \infrule[\ruledef{tvar}]{%
    X \sub T \in \Gamma
  }{%
    \Gamma \ts X \sub T}

  \infrule[\ruledef{capt}]{%
    \Gamma \ts C_1 \sub C_2 \gap
    \Gamma \ts U_1 \sub U_2 
  }{%
    \Gamma \ts \Capt{C_1}{U_1} \sub \Capt{C_2}{U_2}}

  \begin{center}\rule{12cm}{0.4pt}\end{center}

  \infax[\ruledef{refl-pretype}]{%
    \Gamma \ts U \sub U}
  \infrule[\ruledef{trans-pretype}]{%
    \Gamma \ts U \sub V \gap \Gamma \ts V \sub W
  }{%
    \Gamma \ts U \sub W}

  \infax[\ruledef{top}]{%
    \Gamma \ts U \sub \top}

  \infrule[\ruledef{fun}]{%
    \Gamma \ts S_2 \sub S_1 \gap \Gamma \ccomma x: S_2 \ts T_1 \sub T_2
  }{%
    \Gamma \ts \Tlam{x}{S_1}{T_1} \sub \Tlam{x}{S_2}{T_2}}

  \infrule[\ruledef{tfun}]{%
    \Gamma \ts S_2 \sub S_1 \gap \Gamma \ccomma X <: S_2 \ts T_1 \sub T_2
  }{%
    \Gamma \ts \Ttlam{X}{S_1}{T_1} \sub \Ttlam{X}{S_2}{T_2}}

%

}

\FloatBarrier
\hypertarget{typing-rules}{%
\subsection{Typing Rules}\label{typing-rules}}

\boxfig{fig:typing}{Typing rules of the \CC calculus.}{
  \judgement{\textsf{Typing}}{\fbox{$\Gamma \ts t \typ T$}}

  \infrule[\ruledef{var-concrete}]{%
    x: \Capt{C}{U} \in \Gamma
  }{%
    \Gamma \ts x \typ \Capt{\{x\}}{U}}
  \infrule[\ruledef{var-tvar}]{%
    x: X \in \Gamma
  }{%
    \Gamma \ts x \typ X}
  \begin{center}\rule{12cm}{0.4pt}\end{center}

  \infrule[\ruledef{sub}]{%
    \Gamma \ts t \typ T
    \gap
    \Gamma \ts T \sub S
  }{%
    \Gamma \ts t \typ S }

  \infrule[\ruledef{abs}]{%
    \Gamma \ccomma x: S \ts t \typ T \gap \Gamma \ts \Tlam{x}{S}{T} \wf 
  }{%
    \Gamma \ts \lam{x}{S}{t} \ \typ\  \Capt{C}{\Tlam{x}{S}{T}} \\
    \kw{where} C = \fv(\lam{x}{S}{t})
  }


  \infrule[\ruledef{app}]{%
    \Gamma \ts t \typ \Capt{C}{\Tlam{x}{S}{T}} \gap \Gamma \ts s \typ S
  }{%
    \Gamma \ts \app{t}{s} \typ \subs{x}{\cv(S, \Gamma)}{T}}

  \infrule[\ruledef{t-abs}]{%
    \Gamma \ccomma X <: S \ts t \typ T \gap \Gamma \ts \Ttlam{x}{S}{T} \wf
  }{%
    \Gamma \ts \tlam{X}{S}{t} \ \typ\  \Capt{C}{\Ttlam{X}{S}{T}} \\
    \kw{where} C = \fv(\tlam{X}{S}{t})
  }

  \infrule[\ruledef{t-app}]{%
    \Gamma \ts t \typ \Capt{C}{\Ttlam{X}{R}{T}} \gap \Gamma \ts S \wf\\
    \Gamma \ts S \sub R
  }{%
    \Gamma \ts \tapp{t}{S} \typ \subs{X}{S}{T} }

%
%
%
%

}

There are four major differences between \fsub typing rules and typing rules for \CC, described
in Figure \ref{fig:typing}.

\subsubsection*{Capture sets on function values}
The \ruleref{abs} and \ruleref{t-abs} rules augment the
result type of the abstracted function with all variables
that are free in the abstracted term; the type of a value $v$ well-typed in $\Gamma$ is of the form $\fv(v) \; U$,
that is the pretype $U$ annotated with the capture set $\fv(v)$.
Observe here that $\cv(\fv(v) \; U) = \fv(v)$, and in general,
for a term $t$ of type $T$ reducing to a value $v$ we have that $\Gamma \ts \fv(v) \sub \cv(T, \Gamma)$.
This is made formal in Section \ref{metatheory} and by Corollary \ref{lemma:predict-term}.
Once again, note that one may immediately drop pure variables from that capture
set by applying subtyping and rule \ruleref{sc-var}.

\subsubsection*{Application}
In rule \ruleref{app}, the result of the function application is the result
type of the function where the bound variable $x$ is substituted with the
capture set of the argument type $S$. This resembles function application for dependent
function types except that the dependencies are restricted to variable tracking.
The capture set $C$ of the function $t$ itself is discarded in an application.

\subsubsection*{Split variable typing rules}
Our calculus has two different rules for typing variables, depending on whether
a variable $x$ is bound to a concrete type $C \; U$ or to a type variable
$X$ in the environment. Intuitively, the capture set of the variable should be
the variable itself, which is indeed the case if it is bound to a concrete type.
This is not only intuitive, but also a desirable property -- for example,
consider that the type of the term $\LAM{x : \UC \; \top} x$ should be
$\TLAM{x : \UC \; \top} \CAPT{x} \top$, reflecting that the capture
set of the returned value is the same as the capture set of the argument passed
in as $x$.
However, since we may not further annotate a type variable with a capture set,
the type of $\LAM{x: X} x$ cannot be $\TLAM{x : X} \CAPT{x} X$ and has to
be $\TLAM{x : X} X$. Accordingly, we have a separate rule for typing term
variables bound to type variables.

\subsubsection*{Well-formedness constraints} Both \ruleref{abs} and \ruleref{t-abs}
  explicitly require the types they assign to terms to be well-formed. We
  discuss this in the following section.

\hypertarget{well-formedness}{%
\subsection{Well-formedness}\label{well-formedness}}

\FloatBarrier

\boxfig{fig:well-formedness}{Well-formedness of types in the \CC calculus -- term variables are only allowed to occur in covariant positions. }{
  \judgement{Well-formedness}{\fbox{$\Gamma \csemi A_+ \csemi A_- \ts T \wf$}}

  \infrule[\ruledef{capt-wf}]{%
    \highlight{C \subseteq A_+} \gap 
    \forall x_i \in C. \; x_i : S_i \in \Gamma \gap
    \Gamma \csemi A_+ \csemi A_- \ts U \wf
  }{%
    \Gamma \csemi A_+ \csemi A_- \ts \Capt{C}{U} \wf }

  \infrule[\ruledef{universe-wf}]{%
    \Gamma \csemi A_+ \csemi A_- \ts U \wf
  }{%
    \Gamma \csemi A_+ \csemi A_- \ts \Capt{\UC}{U} \wf }

  \infrule[\ruledef{tvar-wf}]{%
    X <: T \in \Gamma
  }{%
    \Gamma \csemi A_+ \csemi A_- \ts X \wf }

  \begin{center}\rule{12cm}{0.4pt}\end{center}

  \infrule[\ruledef{fun-wf}]{%
    \Gamma \csemi A_- \csemi A_+ \ts S \wf
    \gap
    \highlight{\Gamma \ccomma x: S \csemi A_+ \cup \{x\} \csemi A_- \ts T \wf}
  }{%
    \Gamma \csemi A_+ \csemi A_- \ts \Tlam{x}{S}{T} \wf }

  \infrule[\ruledef{tfun-wf}]{%
    \Gamma \csemi A_- \csemi A_+ \ts S \wf
    \gap
    \highlight{\Gamma \ccomma X <: S \csemi A_+ \csemi A_- \ts T \wf}
  }{%
    \Gamma \csemi A_+ \csemi A_- \ts \Ttlam{X}{S}{T} \wf }

  \infax[\ruledef{top-wf}]{%
    \Gamma \csemi A_+ \csemi A_- \ts \Top \wf }

%
}

In \fsub, a type is well-formed simply if all type variables mentioned in it are
bound in the environment. Our corresponding judgment is more complicated: it
also tracks the variance at which term variables appear in capture sets embedded within a type.

We need this restriction because of a difference between evaluation and typing.
When typing term application with \ruleref{app}, we substitute the
lambda's parameter $x$ with the $\cv$ of the argument's type $S$ -- indeed, there is not much else we can do.
In contrast, when reducing application
with rule \ruleref{beta-v}, we substitute $x$ with the {\it free variables} of the argument.
However, we only have that $\Gamma \ts \fv(v) \sub \cv(S, \Gamma)$ -- this subcapturing relation may be strict.
There are two ways to think about this fact that we have found intuitive.
One is that the capture set of the argument's type can be widened through subtyping and subcapturing; another is that the capture set of the argument's type is term-dependent, and hence
can shrink under reduction.
To illustrate this, let us consider the term:
\[\id{f} = \LAM{x : \UC \; U} \LAM{y : \CAPT{x} U} y\]
applied to a pure value $v$ of type $\emptycs \; U$.  Notice that $x$ occurs contravariantly in the capture set of parameter $y$. By \ruleref{beta-v}, $\id{f}(v)$ reduces
to $\LAM{y : \CAPT{} U} y$, with type $\CAPT{} \TLAM{y : \CAPT{} U} \CAPT{y} U$.
However, by applying the subtyping rule \ruleref{capt}, we may also assign $v$ the type $\UC \; U$,
and hence type the application $\id{f}(v)$ with the type $\CAPT{} \TLAM{y : \UC \; U} \CAPT{y} U$.
This is unsound, as the function type $\CAPT{} \TLAM{y : \CAPT{} U} \CAPT{y} U$ is categorically {\bf not}
a subtype of $\CAPT{} \TLAM{y : \UC \; U} \CAPT{y} U$; it can be applied to strictly fewer
values.

This motivates our well-formedness judgement, shown in Figure \ref{fig:well-formedness}, which is defined over a triple $\Gamma \csemi A_+ \csemi A_- \ts T \wf$.
Here, $\Gamma$ is the standard environment and $A_+$ and $A_-$ are sets of term variables.  A term variable $x$ can appear covariantly
only if it occurs in $A_+$, and contravariantly only if it occurs in $A_-$.
For brevity, we write $\Gamma \csemi A_+ \csemi A_- \ts T \wf$ where $A_+$ and $A_-$
are sets of both term and type variables in place of
$\Gamma \csemi A_+ \cap D \csemi A_- \cap D \ts T \wf$ where $D$ is the set of
term variables bound in $\Gamma$.  We also write $\Gamma \ts T \wf$
in place of $\Gamma \csemi \dom(\Gamma) \csemi \dom(\Gamma) \ts T \wf$.
To ensure that subtyping holds
with respect to our term-dependent capture sets, we enforce that a term variable $x$ in a type $T$ can only occur
in covariant position with respect to its binding form in the type by the rules \ruleref{capt-wf}, \ruleref{fun-wf} and \ruleref{tfun-wf}.
This notion is formalized in Section \ref{metatheory}.

As we have seen, the well-formedness condition prevents direct coupling of capture sets
at different polarities. This is less of a restriction than it might seem, since we can
express the same coupling going through a type variable. Here is a version of function
$\mathsf{f}$ that typechecks:
\[ \mathsf{f'} = \tlam{X}{\Capt{\UC}{U}}{\lam{x}{X}{\lam{y}{X}{y}}} \]
Note also that the well-formedness restriction only applies to the variables bound locally
in a type, not to the variables in the global environment. So the function
\[
  \mathsf{g} = \lam{x}{\Capt{\UC}{U}}
    \app{(\lam{y}{\capt{x}{U}}{y})}{x}
\]
is well typed with type $\Tlam{x}{\Capt{\UC}{U}}{\Capt{\{x\}}{U}}$,
even though $x$ is captured at negative polarity in the second lambda.

\FloatBarrier

\hypertarget{metatheory}{%
  \subsection{Metatheory}\label{metatheory}}
We now discuss a few interesting metatheoretic properties of \CC.
The paper is accompanied by a mechanization using the Coq theorem prover, described in more
detail in Section \ref{mechanization}.
We start by observing that \CC is indeed a straightforward extension of \fsub.
In particular, erasing capture sets from well-typed \CC terms yields well-typed \fsub terms.

\begin{lemma}[Erasure]
Let $t$ be a \CC term such that $\ts t : T$ for some type $T$. Let $\lceil\;\cdot\;\rceil$ be a
function from \CC terms and types to \fsub terms and types that erases capture sets (and thereby all term dependencies).
Then we have that $\ts : \lceil t \rceil : \lceil T \rceil$.
\end{lemma}
\begin{proof}
Immediate from structural induction on the typing derivation of $\ts t : T$.
\end{proof}

\noindent Moreover, \fsub embeds naturally into \CC, simply by annotating \fsub function and type
abstraction types with either the empty or the universal capture set.
\begin{lemma}[Embedding]
  Let $C$ be either the empty or the universal capture set.
  Let $t$ be a \fsub term such that $\ts t : T$ for some type $T$.  Let $\lfloor\;\cdot\;\rfloor$ be a
  function from \fsub to \CC terms and types that annotates \fsub types of function and type abstractions with
  $C$. Then we have that $\ts \lfloor t \rfloor : \lfloor T \rfloor$.
\end{lemma}
\begin{proof}
  Structural induction on the typing derivation, after observing that no matter what $C$ is,
  every term variable will be a subcapture of $C$.
\end{proof}

\noindent All of the following lemmas and theorems were mechanized in Coq.

\subsubsection*{Soundness} Our calculus satisfies the standard progress and
preservation lemmas.

\begin{theorem}[Progress]
\label{th:progress}
  If $\ts t : T$, then either $t$ is a value, or there exists a term $t'$ such that we can take a step $t \longrightarrow t'$.
\end{theorem}
\begin{theorem}[Preservation]
\label{th:preservation}
  If $\Gamma \ts t_1 : T$ and $t_1 \reduces t_2$, then we have that $\Gamma \ts t_2 : T$.
\end{theorem}

\subsubsection*{Meaning of capture sets} We observe that the capture set of a
value's type matches the capture set of that value's free variables:

\begin{lemma}[Capture Prediction for Values]
  \label{lemma:prediction}
  If $\Gamma \ts v : T$, then $\Gamma \ts \fv(v) \sub \cv(T, \Gamma)$.
\end{lemma}
\begin{proof}
  Induction on the typing derivation $\Gamma \ts v : T$.
  Now, as $v$ is a value,
  the base case is either an application of the typing rule \ruleref{abs} or \ruleref{t-abs},
  and hence $T = \Capt{\fv(v)}{U}$ for some pretype $U$, as desired.
  Inductively, we have an application of the typing rule \ruleref{sub}.  Hence $\Gamma \ts v : T'$,
  $T' \sub T$, and $\Gamma \ts \fv(v) \sub \cv(T', \Gamma)$.  Now, as $v$ is a value, $T' = \Capt{C'}{U'}$
  for some capture set $C'$ and pretype $U'$, and hence $T = \Capt{C}{U}$ for some capture set $C$ and pretype $U$.
  Hence $\fv(v) \sub \cv(T', \Gamma) = C' \sub C = \cv(T, \Gamma)$, as desired.
\end{proof}

Note that $\fv(v)$ and $\cv(T, \Gamma)$ do not need to be subsets - they need only be in a
subcapturing relationship; for example, consider a value $v = \lam{x}{\Capt{\UC}{\top}}{y}$ in
an environment $\Gamma = (x : \Capt{\emptycs}{\top})$.  Here we may assign $v$ the type
$T = \Capt{\emptycs}{(\Capt{\UC}{\top} \to \top)}$ by subsuming away the capture set for $x$,
but we also have that $\Gamma \ts \fv(v) = \{x\} \sub \{\} = \cv(T, \Gamma)$.

The following corollary captures the essence of \CC. From preservation and
capture prediction for values, it follows that our calculus accurately tracks
the free variables (\ie captured) of the value a term reduces to.

\begin{corollary}[Capture Prediction for Terms]
  \label{lemma:predict-term}
  Let $\Gamma$ be an environment with only term variables.  If $\Gamma \ts t : T$ and $t \longrightarrow^* v$, then $\Gamma \ts \fv(v) \sub \cv(T, \Gamma)$.
\end{corollary}

While the corollary appears deceptively simple, it has important consequences.
In a setting with capabilities, the capture set of a term accurately reflects
the capabilities retained by the value it reduces to.

\subsubsection*{Substitution Lemmas} Due to the term dependency in \CC, we needed to prove a few nonstandard
substitution lemmas for progress and preservation.  This is apparent when comparing the typing rule for term
application with the reduction rule for term application; term substitution proceeds with the exact
capture set of the value -- the free variables of the value, but the typing rule proceeds with a capture
set that subcaptures the free variables of the value.  This necessitates the following lemma, linking
these two capture sets.

\begin{lemma}[Term substitution preserves typing]
  \label{lem:subst_ee'}
  \hfill\\
  \indent If $\Gamma \ccomma x: S \ts t \typ T$ and
  $\Gamma \ccomma x: S \csemi \cset{x} \cup \dom(\Gamma) \csemi \dom(\Gamma) \ts T \wf$,
  then for all $v$ such that $\Gamma \ts v \typ T$,
  we have:
  \[\Gamma \ts \subs{x}{v}{\subs{x}{\fv(v)}t} \typ \subs{x}{\cv(S, \Gamma)}{T}\]
\end{lemma}

\noindent Without the well-formedness condition, we would only be able to show that:

\[\Gamma \ts \subs{x}{v}{\subs{x}{\fv(v)}t} \typ \subs{x}{\fv(v)}{T}\]

\noindent Now, as $\Gamma \ts \fv(v) \sub \cv(S, \Gamma)$, and as $x$ does not occur
contravariantly in $T$ due to our well-formedness constraints,
we have that $\Gamma \ts \subs{x}{\fv(v)}{T} \sub \subs{x}{\cv(S, \Gamma)}{T}$.
Formally, this is stated below in the following lemma, which is needed to prove Lemma \ref{lem:subst_ee'}:

\begin{lemma}[Monotonicity of covariant capture set substitution]
  \hfill\\
  \indent If $\Gamma,x:S \csemi {x} \cup \dom(\Gamma) \csemi \dom(\Gamma) \ts T \wf$,
  then for all $C_1, C_2$ such that $\Gamma \ts C_1 \sub C_2$,
  we have:
  $$\Gamma \ts \subs{x}{C_1}T \sub \subs{x}{C_2}T$$
\end{lemma}

\hypertarget{mechanization}{%
\subsection{Mechanization}\label{mechanization}}
We mechanized \CC using the Coq theorem prover \cite{Coq} \cite{yves04interactive}. In
addition, we wrote a simple typechecker for our terms and used it to verify that
the examples we used in our case studies typecheck correctly. We have also
verified the correctness of this typechecker by proving in Coq that the terms
for $\operatorname{nil}$ and $\operatorname{cons}$ typecheck with the types
given by our simple typechecker.

As our calculus is an extension of \fsub, augmented with sets of free variables
meant to track capture, we based our Coq implementation on the locally nameless
proof of \fsub by \citet{aydemir2008engineering}. In particular, since our types
can mention term variables, we chose the locally nameless approach to avoid
problems with alpha-equivalence of types. We attempted to stay as close as
possible to the original proof of soundness of \fsub. We highlight some of the
details below.

\subsubsection*{Formalizing Capture Sets}
Capture sets in \CC are formalized as an inductive data type with two constructors representing universal
capture sets $\UC$ or concrete capture sets, correspondingly. Due to the locally nameless approach,
a concrete capture set is represented by two sets to model free variables using names and bound
variables using de Bruijn indices.
This worked well for the most part, but we encountered some
difficulties when dealing with sets, as we often had sets that were equal
propositionally, but not definitionally -- for example, $\{1, 2, 3\}$ instead of $\{1\} \cup \{2\} \cup \{3\}$.

\subsubsection*{Formalizing Well-formedness}
As our calculus is dependently typed with respect to capture sets, we need to enforce variance constraints
on where term variables can be bound in a type, as noted in Section \ref{well-formedness}.  Our well-formedness
judgement needs to keep track of two sets of term variables $A_+$ and $A_-$, which describe the
variables in covariant position relative to the current location in the type, and contravariant position respectively.
We modelled this in Coq by defining our inductive well-formedness proposition over a triple $(\Gamma, A_+, A_-)$,
where $\Gamma$ is the classical binding environment, carried over from the \fsub proof, and $A_+$ and $A_-$
are two sets of names.  We found this approach worked well
for describing the modified well-formedness lemmas and also allowed us to prove the necessary weakening and
narrowing lemmas for the overall soundness proof.
In particular, using sets as opposed to lists in the well-formedness judgment allowed us to avoid mechanizing a proof
that well-formedness is preserved under permuting the sets of term variables.
A downside of this representation of well-formedness was that the large number of constraints imposed by
well-formedness conditions made it difficult to formalize example typing judgments.

\section{Language Extensions}
\label{sec:extensions}
The calculus we have presented in the previous section assigns no particular
meaning to capture sets - it merely {\it tracks} the free variables without giving
them any concrete semantics.

This is fully intentional - we believe variable tracking to be a widely
applicable idea and as such, we did not want to privilege any single application
above others by adding it to the base calculus. Instead, in this section we show
how the core calculus can be extended with different semantics for free variables,
and how its metatheory can be used to reason about the extensions.

\subsection{Data Structures in \CC - $\id{List}$\label{sec:list}}
To give some intuition for the calculus, we work out the type signatures of different versions of the $\id{map}$ function, which
maps an arbitrary function argument over a strict list of pure values.  Below, we illustrate type signatures for
the standard $\id{map}$ function and a variant $\id{pureMap}$, which only maps a function that is pure.  $\id{pureMap}$
is of interest in many contexts; for example, one may wish to map a function that possesses no capabilities for performing
side-effects, in order to safely parallelize the map.

We can encode $\id{List}$ in \CC using the standard right-fold
B{\"o}hm-Berarducci encoding \cite{bohmAutomaticSynthesisTyped1985}; we give
terms and typings in Appendix A. All lists are annotated with the empty capture set.
Here is an example type signature for $\id{map}$:
\[\small
\begin{array}[t]{l@{\hspace{2mm}}l@{\hspace{1mm}}l}
\id{map}: & &\{\}\;\forall[A] \\
&\rightarrow &\{\}\;\forall[B] \\
&\rightarrow &\{\}\;\forall(\id{xs}: \id{List}[A]) \\
&\rightarrow &\{\}\;\forall(f: \Capt{\UC}{\Tlam{a}{A}{B}}) \\
&\rightarrow &\id{List}[B]
\end{array}
\]
We use here $\forall[X]$ as an abbreviation for $\forall[X \sub \Capt{\emptycs}{\Top}]$.
The function argument to $\id{map}$ may capture arbitrary capabilities.
However, as $\id{map}$ is strict, that capability is not retained in the result type.
If the list and function arguments are
reversed, the signature of $\id{map2}$ is as follows:

\[\small
\begin{array}[t]{l@{\hspace{2mm}}l@{\hspace{1mm}}l}
  \id{map2}: & &\{\}\;\forall[A] \\
  &\rightarrow &\{\}\;\forall[B] \\
  &\rightarrow &\{\}\;\forall(f: \Capt{\UC}{\Tlam{a}{A}{B}}) \\
  &\rightarrow &\capt{f}{\forall(\id{xs}: \id{List}[A])} \\
  &\rightarrow &\id{List}[B]
\end{array}
\]

\noindent Now, there is an additional capture set $\{f\}$, which reflects the fact that
$\id{map2}(f)$ is a partial application that captures $f$. Finally, here's the signature of $\id{pureMap}$;
recall that $\id{pureMap}$ accepts a function $f$ that must be pure:
\[\small
\begin{array}[t]{l@{\hspace{2mm}}l@{\hspace{1mm}}l}
  \id{pureMap}: & &\{\}\;\forall[A] \\
  &\rightarrow &\{\}\;\forall[B] \\
  &\rightarrow &\{\}\;\forall(\id{xs}: \id{List}[A]) \\
  &\rightarrow &\{\}\;\forall(f: \Capt{\emptycs}{\Tlam{a}{A}{B}}) \\
  &\rightarrow &\id{List}[B]
\end{array}
\]

\noindent  Here, $f$ can only be instantiated with functions that may only capture pure values.

\subsubsection*{Conclusion}
\CC is expressive enough that we can embed {\tt List} into it and assign
accurate types to functions operating on lists. We can express a
capture-polymorphic {\tt map} function, as well as one that only accepts
functions that have captured no free variables. In a setting where side effects are
mediated through capabilities tracked with capture sets, the latter function can
be useful when implementing a parallel {\tt map} function.

One limitation with this encoding in \CC is that {\tt List} can only contain
pure elements. One could specialize the list datatype and the type
variables $A$ and $B$ to work with some fixed, given capture set, up to and
including the universal capture set; however this causes some loss of precision.
In a nutshell, \CC models capture polymorphic operations well, but does not
model capture polymorphic data types as well. We aim to resolve this situation
in a future extension of \CC.

\hypertarget{non-local-returns}{%
\subsection{Non-Local Returns}\label{non-local-returns}}
We now study the applicability of \CC to perform simple effect
checking. The principal idea is that instead of extending the language with
an effect system, we represent the ability to perform an effect with a
capability. If we can guarantee that a capability cannot leave a particular scope,
this model scales to handling exceptions (or effect handlers as we will see in Section \ref{general-algebraic-effects}).
To illustrate the general idea, we start by modeling a language feature, which
is slightly simpler than exceptions: non-local returns. Performing a non-local
return allows us to transfer the control flow to the end of a particular block,
without necessarily being within the lexical scope of that block.
The extension is defined in Figure \ref{fig:return}.

\boxfig{fig:return}{Extending \CC with support for non-local returns.}{
  \judgement{Syntax}{}

  \begin{center}
    \[
    \begin{array}[t]{l@{\hspace{10mm}}l}

      U ::= \dots & \lindent{\bbox{Pretypes}} \\
      \gap \Return[T] & \mbox{return capability} \\
      v ::= \dots & \lindent{\bbox{Values}} \\
      \gap x & \mbox{variables} \\

      t ::= \dots & \lindent{\bbox{Terms}}\\
      \gap \HANDLE x: T \IN t & \mbox{return-able block} \\
      \gap \RETURN{t} s & \mbox{explicit return} \\
      \E ::= \dots & \lindent{\bbox{Evaluation context}} \\
      \gap \HANDLE x: T \IN \E & \\
      \gap \RETURN{\E} t & \\
      \gap \RETURN{x} \E & \\

    \end{array}
    \]
  \end{center}

  \judgement{Reduction}{\fbox{$t \reduces t$}}

  \infax[\ruledefN{beta-return}{beta-return}]{%
    \HANDLE x: T \IN v \reduces v }

  \infax[\ruledefN{context-return}{context-return}]{%
    \HANDLE x: T \IN \E[\RETURN{x} v] \reduces v }

  \judgement{Type assignment}{\fbox{$\Gamma \ts t \typ T$}}

  \infrule[\ruledef{return}]{%
    \Gamma \ccomma x: \UC \; \Return[T] \ts t \typ T
    \gap
    \Gamma, x: \UC \; \Return[T] \nts \cset{x} \sub \cv(T, \Gamma)
  }{%
    \Gamma \ts \HANDLE x: T \IN t \typ T }

  \infrule[\ruledef{do-return}]{%
    \Gamma \ts t \typ \CAPT{C} \Return[T]
    \gap
    \Gamma \ts s \typ T
  }{%
    \Gamma \ts \RETURN{t} s \typ S }
}

\subsubsection*{Operational semantics}

We introduce two new reduction rules and three new evaluation contexts. The latter two of
the three new contexts are standard, but let us pay closer attention to first
one, which mentions $\kw{handle}$. Here, we allow reducing {\it under} a binder for
the return capability. There are two ways for reduction to remove the binder -
either by reducing to a value and applying rule \ruleref{beta-return}, which
corresponds to normally returning from a block; alternatively, the term inside
the block can invoke the return capability and explicitly return from it, which
corresponds to the \ruleref{context-return} reduction rule. Note that if a term
tried to invoke the capability after we have removed the binder from the
evaluation context, the term would be stuck.

\subsubsection*{Soundness}

In order for the semantics of our extension to be sound, the capability to
return from a block should not outlive the block itself. There are two ways it
could do so: either by being returned from it normally (with rule
\ruleref{beta-return}), or by being returned from it explicitly (with rule
\ruleref{context-return}). We prevent both with the non-derivation subcapturing precondition
in rule \ruleref{return}. To see the precise reason why, consider the following.
If returning a value $v$ of type $T$ could leak the capability $x$, then
$x \in \fv(v)$. Then by Lemma
(\ref{lemma:prediction}) and by inspecting the subcapturing rules, it
follows that $\cset{x} \sub \cv(T, \Gamma)$. However, this is forbidden
by \ruleref{return}; it is not possible for a
capability to return from a block to outlive the block itself.

\subsubsection*{Example}

To demonstrate non-local returns, we present a small program that sums up the
square roots of a list of numbers, returning $\id{NaN}$ if one of the numbers is
negative.

\newcommand{\squareRoot}{\id{sqrt}}
\newcommand{\NaN}{\id{NaN}}
\newcommand{\ret}{\local{ret}}
\newcommand{\retCap}{\capability{r}}

{\small
\begin{align*}
&\id{root} : \Double \rightarrow (\UC \; \Double \rightarrow \Double) \rightarrow \Double\\
&\id{root} = \LAM{x} \LAM{\ret}\\
&\qquad \IF x < 0 \THEN{} \ret \NaN \ELSE \squareRoot x\\
&\\
&\id{sumRoots} : \List{\Double} \rightarrow (\UC \; \Double \rightarrow \Double) \rightarrow \Double\\
&\id{sumRoots} \; (x :: xs) = \id{root} x \ret \;+\; \id{sumRoots} xs \ret\\
&\id{sumRoots} \; [] = 0.0\\
&\\
&\HANDLE \retCap: \Double \IN\\
&\qquad \id{sumRoots} \; [1.0, 2.0, 3.0, -1.0] \; (\LAM{x} \RETURN{\retCap} x)
\end{align*}
}

\noindent The program is partitioned into three parts.
Firstly, the $\id{root}$ function takes the square root of
  its argument. If the argument is negative it signals this fact by
  invoking the passed $\ret$ function with the special value $\NaN$ as argument.
Secondly, the $\id{sumRoots}$ function applies $\id{root}$ to each element of a
  list and sums up the results. It simply passes the $\ret$ function to $\id{root}$.
Thirdly, the $\kw{handle}$ expression introduces the $\retCap$ capability. It
  further creates a function that captures the $\retCap$ capability and passes it
  to $\id{sumRoots}$.
The example shows how non-local returns allow transferring the control to a surrounding
handler. Note how the call to $\ret$ in function $\id{root}$ is not in the lexical scope of
the handler that introduces $\retCap$.

The program typechecks since it can be shown that the capability $\retCap$ is not captured by the result of
application of $\id{sumRoots}$. On
the other hand, the following variation gives a type error:

{\small
\begin{align*}
&\HANDLE \retCap: \Unit \rightarrow \Double \IN\\
&\qquad \LAM{} \id{sumRoots} \; [1.0, 2.0, 3.0, -1.0] \; (\LAM{x} \RETURN{\retCap} (\LAM{} x))
\end{align*}
}

\noindent Here, by rule \ruleref{abs}, $\retCap$ does appear in the
capture set of the handler's body, which violates the requirement for
\ruleref{return}.

\subsubsection*{Conclusion} The type system of \CC can indeed support the
notion of non-escaping variables, which we have used in this extension to model
blocks that safely allow non-local returns. We have also seen that functions can
be naturally used in our system to mediate access to capabilities. If our
extensions allowed capability-based exceptions, we would be able to call $\id{sumRoots}$ with
an exception-throwing $\ret$ without any changes to the function's definition.

\hypertarget{region-extension}{
\subsection{Regions}\label{region-extension}}
\boxfig{fig:region-extension}{Extending CC$_{<:}$: with support for region-based memory-management.}{
  \judgement{Syntax}{}
  \begin{center}
  \[
  \begin{array}[t]{l@{\hspace{10mm}}l}
    U ::= \dots & \lindent{\bbox{Pretype}} \\
    \gap \Region & \lindent{\mbox{Region handle}} \\
    \gap \Ptr[T] & \lindent{\mbox{Pointer}} \\
    t ::= \dots & \lindent{\bbox{Terms}}\\
    \gap \region{x}{t} & \lindent{\mbox{Region block}}\\
    \gap \new{x}{T}\;t & \lindent{\mbox{Pointer allocation}}\\
    \gap !\;t & \lindent{\mbox{Pointer de-reference}}\\
    v ::= \dots & \lindent{\bbox{Values}} \\
    \gap x & \mbox{variables} \\
  \end{array}
  \]
  \end{center}

  \judgement{Type assignment}{\fbox{$\Gamma \ts t \typ T$}}
  \infrule[\ruledef{region}]{%
    \Gamma \ccomma x: \Capt{\UC}{\Region} \ts t \typ T
    \gap
    \Gamma, x: \Capt{\UC}{\Region} \nts \cset{x} \sub \cv(T, \Gamma)
  }{%
    \Gamma \ts \region{x}{t} \typ T }

  \infrule[\ruledef{new}]{%
    x: \Capt{C}{\Region} \in \Gamma \qquad \Gamma \ts t \typ T \\
  }{%
    \Gamma \ts \new{x}{T}\;t \typ \capt{x}{\Ptr[T]}}

  \infrule[\ruledef{deref}]{%
    \Gamma \ts t \typ \Capt{C}{\Ptr[T]} \\
  }{%
    \Gamma \ts !\;t \typ T}
}

In another extension, we study the applicability of \CC to region-based memory
management \cite{tofte1997region}. Briefly, the idea of regions is as follows: we can manually allocate
data in regions, which are lexically delimited scopes. We statically ensure that
after a region is left, no reference to data allocated in the region remains,
which means that we can safely deallocate the entire region. As such, this
approach is a natural fit for being expressed with \CC.

We draft the extension in Figure \ref{fig:region-extension}. We assume standard
store-based operational semantics \cite{grossman2002regions}; in particular, we
assume that the value for pointers mentions the region in its free variables.
The overall approach is analogous to the one in the non-local return extension.
We add a binder for regions and reduce under it; the binder introduces a handle
for the region into scope, which can be used to allocate data in the region.
Similar to the non-local return extension (Section \ref{non-local-returns}), if
either the region handle or a pointer allocated in the region leaves the region,
the extension would be unsound. We again prevent this with the non-derivation subcapturing
precondition on rule \ruleref{region}.

\subsubsection*{Conclusion}
\CC can be used to model a discipline for safe memory management as well as
effects. We can define region-polymorphic functions without needing explicit
region polymorphism. As an example, consider the following function, which
simply de-references an arbitrary pointer:

$$
\tlam{Y}{\Capt{\UC}{\Top}}{%
  \lam{y}{\Capt{\UC}{\Ptr[Y]}}{%
    \;!\;y
  }}
$$

\noindent This does not rule out explicitly qualifying functions with regions where
necessary. Consider the following function, which accepts a handle to a region
and a pointer allocated on that region, and duplicates the pointer it received:

$$
\tlam{Y}{\Capt{\UC}{\Top}}{%
  \lam{x}{\Capt{\UC}{\Region}}{%
    \lam{y}{\capt{x}{\Ptr[Y]}}{%
      \new{x}{Y}(!\;y)
    }}}
$$

\noindent We can use the capture sets of functions to reason about the regions that they can
access. In particular, we can know which regions they cannot possibly access.

\hypertarget{general-algebraic-effects}{%
\subsection{Effect Handlers}\label{general-algebraic-effects}}

As a final case study, we generalize the system of non-local returns to algebraic effects and handlers \cite{plotkin2003algebraic, plotkin2013handling}.
Effect handlers are a program structuring paradigm that allows to model complex control-flow patterns in a structured way.
We build our presentation on effect handlers in capability-passing style \cite{brachthaeuser2017effekt, zhang2019abstraction, brachthaeuser2020effects},
since it perfectly fits our framework of reasoning about free variables and binders.
To keep the presentation of the calculus simple, we follow \citet{zhang2019abstraction} and
limit our effect handlers to only a single operation and no return clauses.

\boxfig{fig:algeff-syntax}{Extended syntax of \CC with support for algebraic effect handlers.}{
  \judgement{Syntax}{}

  \begin{center}
    \[
    \begin{array}[t]{l@{\hspace{10mm}}l}

      U ::= \dots & \lindent{\bbox{Pretypes}} \\
      \gap \Eff[A, B] & \mbox{effect capabilities} \\
      v ::= \dots & \lindent{\bbox{Values}} \\
      \gap x & \mbox{variables} \\
      t ::= \dots & \lindent{\bbox{Terms}}\\
      \gap \HANDLE x: \Eff[A, B] = \LAM{y \; k} s \IN t & \mbox{effect handling} \\
      \gap \app{\eff{x}}{t} & \mbox{effect operation calls} \\
      \E ::= \dots & \lindent{\bbox{Evaluation contexts}} \\
      \gap \HANDLE x: \Eff[A, B] = \LAM{y \; k} s \IN \E & \\
      \gap \app{\eff{x}}{\E} & \\

    \end{array}
    \]
  \end{center}

  }

\autoref{fig:algeff-syntax} extends the basic \CC calculus with additional
syntax for effect handling.
To type capabilities, we add a new pretype $\Eff[A, B]$ that represents effect
operations from $A$ to $B$.
That is, the type parameter $A$ indicates the type of values passed to an effect
operation and type $B$ indicates the type of values returned by an effect operation.
There are two new forms of expressions:
First, the expression $\HANDLE x: \Eff[A, B] = \LAM{y \; k} s \IN t$ acts as a binder and
introduces a capability $x : \UC \; \Eff[A, B]$ in the handled program $t$. The handler
implementation $\LAM{y \; k} s$ has two parameters. Parameter $y$ will be bound to
the argument of type $A$ passed to the effect operation. Parameter $k$
represents the continuation. To avoid having to annotate the type of the
continuation, we slightly diverge from our notation of function binders here,
since the type annotation on $x$ suffices. We also sometimes use the shorthand
$\HANDLE x = h \IN s$.
Second, within the handled program $t$, calling an effect operation with $\app{\eff{x}}{v}$
suspends the current computation, passing the argument $v$ to the handler bound to $x$.

Our description of the operational semantics of handlers closely follows the
\emph{open semantics} presented by \citet{biernacki2020binders}. In this style, effect
handlers are treated as binders for capabilities. Effect operations are
reduced by evaluating \emph{under} those binders, while preserving the usual call-by-value
left to right evaluation strategy for all other abstractions.
As a consequence, like with non-local returns, we add variables to the syntactic category of values.
This way, capability references can be passed as arguments to functions.
Treating effect handlers as binders is a perfect fit for \CC,
since the core idea of \CC is to track free variables in the type of abstractions --
equally relying on lexical binding.

\boxfig{fig:algeff-evaluation}{Extended operational semantics of \CC with support for algebraic effect handlers.}{
  \judgement{Reduction}{\fbox{$t \reduces t$}}

  \infax[\ruledefN{beta-handle}{beta-handle}]{%
    \HANDLE x = h \IN w \reduces w
  }

  \infrule[\ruledef{context-handle}]{%
    \E = \HANDLE x = \LAM{y \; k} s \IN \E'
  }{%
    \E[\app{\eff{x}}{v}] \reduces \subs{k}{\LAM{z} \E[z]}{\subs{y}{v}{s}}
  }
}

\subsubsection*{Operational Semantics}

There are two new reduction rules. The first rule removes a handler abstraction if the program $w$ is already evaluated to a value. Importantly, this is only safe when $w$ does not contain $x$ free. As we will see, our extended typing rules prevent this source of unsoundness.
The second rule connects effect operation calls on $x$ with the corresponding handler binding it.
To reduce an effect call $\app{\eff{x}}{v}$ in a context $\E$, the context needs to provide a handler for $x$.
Furthermore, the evaluation context between the handler and the effect operation call is denoted by $\E'$.
We evaluate the effect operation call by substituting the argument $v$ for $y$,
and the continuation $\LAM{z} \E[z]$ for $k$ into the handler body $s$.
Calling the continuation will reinstantiate the delimited evaluation context $\E$ that also contains the handler binding $x$.
Our operational semantics thus implements \emph{deep handlers} \cite{kammar2013handlers}.

\subsubsection*{Typing Rules}

\boxfig{fig:algeff-typing}{CC$_{<:}$: algebraic effect extension typing rules}{
  \judgement{Type assignment}{\fbox{$\Gamma \ts t \typ T$}}


  \infrule[\ruledef{handle}]{%
    \textit{(1a)\quad} \Gamma \ccomma x: \UC \; \Eff[A, B] \nts \cset{x} \sub \cv(A, \Gamma) \\
    \textit{(1b)\quad} \Gamma \ccomma x: \UC \; \Eff[A, B] \nts \cset{x} \sub \cv(R, \Gamma) \\
    \textit{(2)\quad} \Gamma \ccomma y : A \ccomma k : C_k \: B \rightarrow R \ts s \typ R\ \qquad\qquad
    C_k = (\fv(t) \setminus \{x\}) \cup (\fv(s) \setminus \{y, k\}) \\
    \textit{(3)\quad} \Gamma \ccomma x: \UC \; \Eff[A, B] \ts t \typ R \\
  }{%
    \Gamma \ts \HANDLE x: \Eff[A, B] = \LAM{y \; k} s \IN t \typ R}

  \infrule[\ruledef{do}]{%
    \Gamma \ts x \typ C \; \Eff[A, B]\\
    \Gamma \ts t \typ A
  }{%
    \Gamma \ts \app{\eff{x}}{t} \typ B
  }
}

The typing rules for general effect handlers are naturally more complex than the ones for
non-local returns, but a core principle stays the same: In both cases the \ruleref{handle} rule requires that the locally defined handler $x$ does not escape in the handled expression's result.
We can group the premises into two categories:
The first two rows of premises {\it (1a)} and {\it (1b)} are well-formedness conditions to assert non-escaping.
The other two rows of premises type check the handler body \textit{(2)} and the handled program \textit{(3)}.
Starting from the last premise, we will now work through the different premises, highlighting important aspects.
Premise \textit{(3)} type checks the handled program and brings
a capability of type $\Eff[A, B]$ into scope. By annotating it with the universal
capture set, we mark the capability as tracked.
Premise \textit{(2)} types the body of the handler. It not only binds the argument of the
effect operation $y$, but also the continuation, to which we assign the type
$C_k \; B \rightarrow R$. Interestingly, the type expresses that the continuation captures
exactly the union of free variables of our handled program and the free variables of the handler.
Finally, to guarantee that capabilities cannot escape,
premises {\it (1a)} and {\it (1b)} require that the singleton capture set $\cset{x}$ is not a subcapture
of $\cv(A, \Gamma)$ (and $\cv(R, \Gamma)$ respectively). This has an interesting consequence:
the capture sets of $A$ and $R$ need to be concrete capture sets -- they cannot
be the universal capture set, since then subcapturing would hold.
This restriction lets us rule out programs such as:
\newcommand{\xcap}{\capability{x}}
\newcommand{\ycap}{\capability{y}}
\begin{align*}
&\HANDLE \capability{x} = v \IN%
  \LAM{y} \app{\eff{\xcap}}{y}%
\qquad\reduces\qquad%
\LAM{y} \app{\eff{\xcap}}{y}
\end{align*}

\noindent where $\xcap$ is unbound after reduction.

In addition to restricting the answer type $R$, we also restrict the argument type $A$. The motivation for this is more subtle.
Let us assume the following example adapted from \citet{biernacki2020binders}:

\newcommand{\thunkVar}{\local{thunk}}

\begin{align*}
&\HANDLE \xcap: \Eff[\UC \; \Unit \rightarrow \Unit, \Unit] = \LAM{\thunkVar \; k} \thunkVar () \IN\\
&\qquad\HANDLE \ycap = h \IN\\
&\qquad\qquad\app{\eff{\xcap}}{ \LAM{} \app{\eff{\ycap}}{()} }
\end{align*}

\noindent The example reduces in the following way
\begin{align*}
&\subs{\thunkVar}{...}{\subs{k}{...}{(\app{\thunkVar}{()})}}\\
&\qquad\reduces\\
&\app{(\LAM{} \app{\eff{\ycap}}{()})}{()} \\
&\qquad\reduces\\
&\app{\eff{\ycap}}{()}
\end{align*}

\noindent again leading to an unbound, that is unhandled, effect call on $\ycap$. To avoid this, we need to rule out the possibility that
lambda abstractions closing over capabilities at the call site can be passed to effect operations.
By requiring that the capture set on $A$ needs to be concrete, we rule out the type of
\[
\Eff[\UC \; \Unit \rightarrow \Unit, \Unit]
\]
\noindent instead we would need to give the more precise type
$\Eff[\CAPT{\ycap} \Unit \rightarrow \Unit, \Unit]$. However, this is again ruled out, since it is not well-formed in the outer typing context. $\ycap$ is not bound at the handling site of $\xcap$.

\subsubsection*{Conclusion}
Capture sets allow us to reason about capability safety: without equipping the language with an
additional effect system, we can be sure that all effects are handled simply by establishing that
capabilities do not leave their corresponding effect handlers. Capture sets also allow us to reason
about the effects used by a function. Inspecting the capture set on the type of a function value, we
can conclude which effects can potentially be used by this function and in particular, which effects
\emph{cannot} be used.

\hypertarget{related-work}{%
\section{Related Work}\label{related-work}}

The key distinction between our approach and similar work in the
literature is that our calculus is \emph{descriptive} rather than \emph{prescriptive}. That
is, our calculus can be understood as tracking aliasing with types, through
which we can express many different concepts. Broadly speaking, other approaches
such as ownership systems, linear types or borrowing use types to restrict some
terms to follow a concrete aliasing hygiene.

Related literature ranges from object capabilities, effect systems, algebraic effects and handlers, to
region-based memory management. Here we offer a comparison to the work that we believe is closely
related.

\subsubsection*{Second-Class Values}
Motivated by goals very similar to our work,
\citet{osvald2016gentrification} present a type-based escape analysis \cite{hannan1998escape} that
allows the tracking of capabilities and prevents them from escaping. They achieve this by distinguishing
between first-class values and \emph{second-class values}. First-class values can be passed to,
returned from, and closed over by functions. In contrast, second-class values are restricted in that they
can never be returned from functions and can only be closed over by other second-class values.
This distinction is an elegant and simple solution that can also encode borrowing \cite{osvald2017rust} and
enables a lightweight form of effect polymorphism \cite{brachthaeuser2020effects}.
However, what makes their calculus so simple also makes it restrictive: Second-class values cannot be returned
under any circumstances, even when this would be sound. In our present work, we relax this restriction by
generalizing first and second-class values to accurately track the captured variables in the type.
First-class values and types are translated into \CC terms and types annotated with the
empty capture set $\emptycs$. That is, they can freely be passed to, closed over, and returned from all other functions.
Second-class values and types are translated into \CC terms and types annotated with the
\emph{universal} capture set $\UC$. That is, they are tracked and cannot be returned or closed over by first-class functions.
Yet, they can close over other second-class values annotated with the universal capture set.

\subsubsection*{Effect Systems}
Effect systems extend the static guarantees of type systems to additionally describe the side-effects a computation may perform \cite{lucassen1988polymorphic}.
This enables programmers to reason about purity and perform semantics-preserving refactorings and security analysts to determine the privileges required by
a computation to be executed. While the \CC calculus can be used to achieve effect safety, there is an important difference to traditional effect systems.
Effect systems typically track the \emph{use} of effectful operations, while in \CC we track the \emph{mention} of resources / capabilities \cite{gordon2020designing}.
This manifests in two ways.

First, typing in the \CC calculus is about \emph{values}, while typing in effect systems is about side-effecting \emph{expressions}. This becomes visible in Lemma \ref{lemma:prediction},
which relates the free variables of a value with the capture set in its types. Let us assume the following example expression
\newcommand{\abortCap}{\capability{abort}}
\begin{align*}
t : \CAPT{} \Unit
t = \app{\abortCap}{()}; ()
\end{align*}

\noindent that is a call to $\abortCap$ followed by returning the unit value.
Effect systems would register the call to $\abortCap$ in the type of the expression, while the \CC calculus assigns it the type $\CAPT{} \Unit$.

This might seem counter-intuitive at first, but is resolved  by the second difference
with effect systems: Reasoning with \CC is about the \emph{context}, while reasoning with effect systems is about \emph{programs}. Since the context includes a binding for $\abortCap$
(that is, \emph{the capability is in scope}) we take it for granted that the expression can use it. Delaying a computation with a (type or term) abstraction
externalizes the dependencies on the context and we obtain: $\LAM{} t \;:\; \CAPT{\abortCap} () \rightarrow \Unit$. Since we only track the dependencies on the
context (that is, mention) and not the use of effect operations, we assign the exact same type to $t' = \LAM{} \abortCap; ()$. In contrast, traditional effect systems
would assign a pure type to $t'$ since it is observationally equivalent to $\LAM{} ()$.
In consequence, while effect systems suggest to reason about purity, in \CC it makes sense to reason about \emph{contextual purity} \cite{brachthaeuser2020effects}.
Delaying computation allows to partially navigate between the two modes of reasoning.

\subsubsection*{Capabilities}
In the (object-)capability model of programming \cite{crary1999typed, boyland2001capabilities, miller2006robust},
performing security critical operations requires access to a \emph{capability}. Such a capability
can be seen as the constructive proof that the holder is entitled to perform the
critical operation. Reasoning about which operations a module can perform is reduced
to reasoning about which references to capabilities a module holds.

The Wyvern programming language \cite{melicher2017capability}, embraces this mode of reasoning
and establishes authority safety by restricting access to capabilities. The language
distinguishes between stateful \emph{resource modules} and \emph{pure modules}. Access to
resource modules is restricted and only possible through capabilities. Determining the authority
granted by a module amounts to manually inspecting its type signature and all of the type signatures of its
transitive imports. To support this analysis, \citet{melicher2020controlling} extends the language with a
fine-grained effect system that tracks access of capabilities in the type of methods.
The extended language supports effect abstraction via abstract effect members, which can also be
bounded \cite{fish2020case} to integrate well with the structural subtyping of Wyvern.
Using this effect system, \citet{melicher2020controlling} formalizes the authority of a module
by collecting the set of effects annotated on methods and transitively of all modules returned
by those methods.

In the \CC calculus, reasoning about authority and capability safety is very similar.
However, access to capabilities is immediately recorded in the capture set.
Modelling modules via function abstraction, the capture set of a function directly reflects the authority of that function.
As an important difference, the \CC calculus does not include an effect system and thus tracks \emph{mention} rather than \emph{use}.
Wyvern allows effect abstraction to be expressed directly via (abstract) effect members on modules; we envision that
\CC can express an analogous form of effect abstraction indirectly via term abstraction and capture polymorphism,
similarly to how existential quantification can be encoded using universal quantification.

\subsubsection*{Coeffects}
Effect systems can be understood as tracking additional information about the output of a typing judgement
$\Gamma \ts e \typ \highlight{\varepsilon} \; \tau$. Dually, coeffect systems \cite{petricek2014coeffects}
equip the context in which an expression is typechecked with additional structure $\Gamma \highlight{@ \; \mathcal{C}} \ts e \typ \tau$.
\citet{petricek2014coeffects} show that coeffects can be instantiated to express linearity of resources, implicit parameters, and many more.
Very similarly, the capture set on term and type abstractions expresses requirements about the context in which these abstractions can be executed.
While \citeauthor{petricek2014coeffects} present a very general framework that can be instantiated with many different use cases, their work is
based on simply typed lambda calculus. In contrast, in the present paper we embrace subtyping (alongside with all its advantages and challenges)
that arises from the notion of capture sets and base our calculus on \fsub.

With the goal to retrofit existing impure languages with a mechanism to reason about purity,
\citet{choudhury2020recovering} introduce a calculus that distinguishes between safe (that is, pure)
terms and impure terms. A special type $\square T$ witnesses that the term cannot close over
any impure bindings, that is over potentially effectful resources. The type comes with an introduction
form  $\BOX e$, which type checks the expression $e$ in a context that only contains pure bindings,
and an elimination form $\LET x = e \IN e_2$, which introduces a pure binding in the context.
Similarly, our capture sets serve as a type-level certificate that the value only closes over
those tracked bindings mentioned in the capture set. To facilitate the comparison with the work
by \citeauthor{choudhury2020recovering}, we can conceptually rephrase our typing rule for abstractions to:

\infrule[\ruledef{abs-filter}]{%
  \highlight{\Gamma^{C}} \ccomma x: S \ts t \typ T
}{%
  \Gamma \ts \LAM{x: S} t \typ \highlight{C} \; \TLAM{x : S} T }

\noindent This rule filters all bindings from the typing context $\Gamma$ that are not captured by $C$. Furthermore,
a binding of type $x : \Capt{\UC}{P} \in \Gamma$ corresponds to an impure binding, whereas
$x : \capt{}{P} \in \Gamma$ models a pure binding. As in the work by \citeauthor{choudhury2020recovering},
pure bindings cannot close over impure bindings and our type $\capt{}{P}$ thus corresponds to the purity witness $\square T$.
As such, \CC can very similarly be used to gradually recover
purity in an impure language. Dual to the encoding proposed in the comparison with second-class values,
we can annotate all existing terms and types with the universal capture set and selectively mark those
functions that are pure with the empty capture set.
Furthermore, our calculus not only allows to express pure and impure
bindings, but extends the binary notion of purity to concrete, finite capture sets. While the
system of \citeauthor{choudhury2020recovering} has an appealing simplicity, \CC incorporates subtyping
and a limited form of term dependency, naturally leading to a naturally more complex system with additional well-formedness conditions.

\subsubsection*{Capabilities in Scala}\label{capabilities-in-scala}

Modeling resources as capabilities and passing them explicitly to other modules
can quickly become tedious. The Scala language comes equipped with \emph{contextual
abstractions} that allow programmers to abstract over capabilities without having to pass them
explicitly. This includes type-directed implicit parameters and implicit function types \cite{odersky2017simplicitly}
which have been introduced in Scala 3. Here is an example of how operations tracking exception
capabilities can be modeled in Scala 3\footnote{In practice, one would rather equip the language-defined `try` and `throw` constructs with similar types.}.
\begin{lstlisting}[language=scalaish,basicstyle=\footnotesize\ttfamily]
    class Exc                         // Exception classes
    class DivByZero extends Exc

    class CanRaise[E <: Exc]          // Capability class
    infix type raises[A, E <: Exc] =  // Capability wrapper
      CanRaise[E] ?=> A

    // Basic exception operations
    def handleWith[A, E <: Exc](body: CanRaise[E] ?=> A)(handler: E => A): A = ...
    def raise[E <: Exc](exc: E): CanRaise[E] ?=> Nothing = ...

    def safeDiv(x: Int, y: Int): Int raises DivByZero =
      if y == 0 then raise(DivByZero()) else x / y
\end{lstlisting}
\noindent Here, we use the implicit function type \verb@CanRaise[E] ?=> A@ to represent expressions
that return a value of type \verb@A@ but that also have the capability to
raise an exception of type \verb@E@. That type can be abbreviated by means of the given type alias
to\ \verb@A raises E@\ . Hence, as can be seen in function \texttt{safeDiv},
programmers need not bind or pass the capability explicitly. While very useful
for modeling contextual abstractions,
implicits do not guarantee effect safe usage of capabilities.
It has been proposed to
combine them with second-class values \cite{brachthaeuser2017effekt,osvald2016gentrification}, or an
embedding of an effect system using other advanced type-level machinery of Scala
\cite{brachthaeuser2020effekt}. With \CC, in this paper, we propose another mechanism to
statically guarantee capability safety that is more expressive than second-class values,
and more lightweight than the embedding by \citet{brachthaeuser2020effekt}.
In an imaginary extension of Scala with capture sets, the \verb@handleWith@ operation would
create a local capability of class \verb@CanRaise@ that it passes to its \verb@body@, while
checking that the result of \verb@body@ does not contain the local capability
in its capture set, similar to the technique used in Section~\ref{non-local-returns}.

\subsubsection*{Regions}
Earlier in Section \ref{region-extension} we showed that we could extend \CC
with support for simple, stack-based regions. Here, we compare our extension
with Cyclone \cite{grossman2002regions}, a C-like language
featuring region-based memory management.

Our extension is, in some ways, limited compared to Cyclone. Because \CC does not
support data structures containing tracked references, we cannot stack-allocate
data structures containing pointers. However, we see no reason to believe this
is a fundamental limitation - with an improved version of \CC that does support
impure data structures, the extension should naturally allow data structures to
mention pointers.

Our extension also does not support sub-regioning. The sub-region problem can be
defined as follows: given two regions $x$ and $y$, with $y$ being shorter-lived
(or more nested), can we pass a pointer of type $\capt{x}{\Ptr[T]}$ where a
pointer of type $\capt{y}{\Ptr[T]}$ is expected? We would be able to do so if we
knew that $\{x\} \sub \{y\}$ based on the bounds of $x$; however, since $y$ is
the more nested region, that is not possible. To support sub-regioning, \CC needs
the reverse bound: the ability to know that $\{y\} \sub \{x\}$ based on the bounds of
$y$, i.e. the ability to put lower bounds on capture sets of term variables.

However, our extension as presented already supports simple regions while using
the more widely applicable type system of \CC. In comparison, Cyclone has much
more specialised features. It has a separate concept of region variables $\rho$
and region handles $\kw{region}(\rho)$. Region-polymorphic functions need to be
explicitly qualified with region variables. It tracks the {\it use} of regions
with an effect system, and, to avoid explicit effect polymorphism, defines a
bespoke {\tt regions\_of} type operator. To contrast that with our calculus \CC,
observe that we do not need to introduce an effect system to support regions;
we support region polymorphism without introducing region variables (as have discussed in in Section \ref{region-extension}),
and we do not need to qualify region-polymorphic functions with regions
unnecessarily. Furthermore, we conjecture that one could lift both of the limitations we
discussed previously without introducing any region-specific features to \CC.

\section{Conclusion}
In this paper, we introduced the \CC calculus, a type-theoretic foundation for tracking
free variables. The calculus is a modest addition to \fsub, integrating the tracking
of free variables with subtyping. However, subtyping also required us to equip the calculus with
additional well-formedness conditions to establish soundness. The calculus satisfies interesting
meta-theoretical properties. In particular, capture sets soundly approximate the free variables
captured by a value, giving rise to reasoning about effect safety in terms of capability safety.
We evaluated the practical applicability of the calculus by presenting several language extensions,
each making use of the newly gained expressive power of the type system. Capture polymorphism provides
a uniform way to express region and effect polymorphism. In the future, it would be interesting to
fully implement the calculus in a practical programming language, to further explore the gained expressivity.

\section{Acknowledgements}
We thank the authors and maintainters of Proof General \cite{ProofGeneral2000}
and {\tt company-coq} \cite{CompanyCoq2016}, Coq development environments which
were indispensable when working on this paper.

\bibliography{bibliography}

\appendix
\section{Appendix: Typing $\id{List}$}

We can represent a list using a function that takes a function $g$ and applies it to the elements of the list.
Specifically, $g$ takes an element of the list $v$ and an already accumulated result $s$ and returns a new
accumulated result. The list applies $g$ to the elements of the list in turn to yield a final accumulated result.
Concretely, the type of a list of elements of type $T$ is:
\begin{align*}
\id{List}[T] \equiv
\Capt{\cset{}}{\Ttlam{C}{\Capt{\UC}{\Top}}{\Capt{\cset{}}{\Tlam{g}{\Capt{\UC}{\Tlam{v}{T}{\Capt{\UC}{\Tlam{s}{C}{C}}}}}{\Capt{\cset{g}}{\Tlam{s}{C}{C}}}}}}
\end{align*}

We define an abbreviation for the type of the function $g$:
\begin{align*}
\id{Op}[T,C] \equiv
\Capt{\UC}{\Tlam{v}{T}{\Capt{\UC}{\Tlam{s}{C}{C}}}}
\end{align*}

Then the list type can be abbreviated to:
\begin{align*}
\id{List}[T] \equiv
\Capt{\cset{}}{\Ttlam{C}{\Capt{\UC}{\Top}}{\Capt{\cset{}}{\Tlam{g}{\id{Op}[T,C]}{\Capt{\cset{g}}{\Tlam{s}{C}{C}}}}}}
\end{align*}

The term representing an empty list ignores $g$ and just applies an identity function to the initial accumulated result:
\begin{align*}
\id{nil} \equiv
\tlam{T}{\Capt{\UC}{\Top}}{\tlam{C}{\Capt{\UC}{\Top}}{\lam{g}{\id{Op}[T,C]}{\lam{s}{C}{s}}}}
\end{align*}

The term representing a cons cell first recurses on the tail of the list, and finally applies $g$ to the head:
\begin{align*}
&\id{cons} \equiv\\
&\tlam{T}{\Capt{\UC}{\Top}}{\lam{hd}{T}{\lam{tl}{\id{List}[T]}{\tlam{C}{\Capt{\UC}{\Top}}{\lam{g}{\id{Op}[T,C]}{\lam{s}{C}{\app{\app{g}{hd}}{\left(\app{\app{\tapp{tl}{C}}{g}}{s}\right)}}}}}}}
\end{align*}

We can now implement the $\id{map}$ function from Section~\ref{sec:list} as follows:
\begin{align*}
&\id{map} \equiv\\
&\tlambr{A}{\Capt{\cset{}}{\Top}}{\tlambr{B}{\Capt{\cset{}}{\Top}}{\lambr{xs}{\id{List}[A]}{\lambr{f}{\Capt{\UC}{\Tlam{a}{A}{B}}}{\app{\app{\tapp{xs}{\id{List}[B]}}{\lam{elem}{A}{\lam{accum}{\id{List}[B]}{\app{\app{\tapp{\id{cons}}{B}}{\left(\app{f}{elem}\right)}}{accum}}}}}{\left(\tapp{\id{nil}}{B}\right)}}}}}
\end{align*}

The $\id{map2}$ function, which swaps the order of $f$ and $xs$, can be implemented as follows with the same
function body:
\begin{align*}
&\id{map2} \equiv\\
&\tlambr{A}{\Capt{\cset{}}{\Top}}{\tlambr{B}{\Capt{\cset{}}{\Top}}{\lambr{f}{\Capt{\UC}{\Tlam{a}{A}{B}}}{\lambr{xs}{\id{List}[A]}{\app{\app{\tapp{xs}{\id{List}[B]}}{\lam{elem}{A}{\lam{accum}{\id{List}[B]}{\app{\app{\tapp{\id{cons}}{B}}{\left(\app{f}{elem}\right)}}{accum}}}}}{\left(\tapp{\id{nil}}{B}\right)}}}}}
\end{align*}

Finally, the $\id{pureMap}$ function also has the same function body, but the parameter type for the function $f$ enforces that this function is pure:
\begin{align*}
&\id{pureMap} \equiv\\
&\tlambr{A}{\Capt{\cset{}}{\Top}}{\tlambr{B}{\Capt{\cset{}}{\Top}}{\lambr{xs}{\id{List}[A]}{\lambr{f}{\Capt{\cset{}}{\Tlam{a}{A}{B}}}{\app{\app{\tapp{xs}{\id{List}[B]}}{\lam{elem}{A}{\lam{accum}{\id{List}[B]}{\app{\app{\tapp{\id{cons}}{B}}{\left(\app{f}{elem}\right)}}{accum}}}}}{\left(\tapp{\id{nil}}{B}\right)}}}}}
\end{align*}

We have constructed typing derivations for all of these terms to make sure that they have the claimed types.
\end{document}